\documentclass[11pt,a4paper]{article}
\usepackage{lipsum}
\usepackage{authblk}
\usepackage{amsmath,a4}  
\usepackage[utf8]{inputenc}
\usepackage{graphicx,enumerate,color,amsfonts,amsthm,amsmath,amssymb,mathrsfs,marginnote,dsfont,textcomp,tikz,marginnote,subfigure,pdflscape,multirow,booktabs,float,enumitem,scrtime,currfile,titleps,hyphenat,algorithm,algpseudocode,bigints,diagbox,multirow,mathtools,mdframed}
\usepackage[colorlinks,
    linkcolor={red!60!black},
    citecolor={black},
    urlcolor={blue!80!black}]{hyperref}
\usepackage{tikz,pgf,pgffor,pgfplots}
\tikzset{>=latex}
\usetikzlibrary{calc,patterns,decorations,intersections,decorations.pathmorphing,fit,backgrounds,arrows,shapes,snakes,automata,petri,positioning,,graphs,graphs.standard}
\usepgfmodule{shapes,plot}

\newtheorem{theorem}{Theorem}[section]

\theoremstyle{definition}

\theoremstyle{remark}
\newtheorem*{remark}{Remark}

\makeatletter
\newcommand{\setword}[2]{%
  \phantomsection
  #1\def\@currentlabel{\unexpanded{#1}}\label{#2}%
}
\makeatother

\newcommand{\dd}{\textnormal{d}}

\newcommand{\R}{\mathbb{R}}

\newcommand{\RM}{\mathcal{R}}

\newcommand{\abs}[1]{\left\lvert #1\right\rvert}

\newcommand{\set}[1]{\left\{#1\right\}}

\allowdisplaybreaks
\usepackage[top=2cm, bottom=2cm, left=2cm, right=2cm]{geometry}
\usepackage{fancyhdr}
\newpagestyle{main}{
\setheadrule{.4pt}
\sethead{}
        {}
        {\sectiontitle}
}
\renewenvironment{abstract}{%
\hfill\begin{minipage}{0.95\textwidth}
\rule{\textwidth}{1pt}}
{\par\noindent\rule{\textwidth}{1pt}\end{minipage}}
\makeatletter
\renewcommand\@maketitle{%
\hfill
\begin{minipage}{0.95\textwidth}
\vskip 2em
\let\footnote\thanks 
{\Large \bf \@title \par }
\vskip 1.5em
{\large \@author \par}
\end{minipage}
\vskip 1em \par
}
\makeatother
\begin{document}
%
\title{\nohyphens{An optimal control model of mosquito reduction management in a dengue endemic region}}
\author[a,*]{Karunia Putra Wijaya}
\author[a]{Thomas G\"{o}tz}
\author[b]{Edy Soewono}
\affil[a]{\small\emph{Mathematical Institute, University of Koblenz, 56070 Koblenz, Germany}}
\affil[b]{\small\emph{Department of Mathematics, Bandung Institute of Technology, 40132 Bandung, Indonesia}}
\affil[$\ast$]{Corresponding author. Email: \href{mailto:karuniaputra@uni-koblenz.de}{karuniaputra@uni-koblenz.de}}
\maketitle
\begin{abstract}
{\textbf{Abstract}}: \emph{Aedes aegypti} is known as the responsible vector transmitting dengue flavivirus. Unavailability of medication to cure the transmission of the virus in the human blood becomes a global health issue in recent decades. World epidemiologists are encouraged to focus on the investigation over the effective and inexpensive way to prevent dengue transmission, i.e. mosquito control. In this paper, we present a model depicting the dynamics of mosquito population based on indoor-outdoor life cycle classification. The basic mosquito offspring number was obtained and analysis of equilibria was shown. We brought along a discussion on the application of optimal control to the model in which two simultaneous schemes were introduced. The first scheme is done by disseminating chemical like temephos in spots where eggs and larvae develop, meanwhile the second scheme is done by deploying fumigation through areas where adult mosquitoes prevalently nest, indoor as well as outdoor. A version of the gradient-based method was presented to set up a workflow in minimizing the objective functional with respect to some control variables. Numerical results from the analysis of the basic mosquito offspring number with constant control and from that with optimal control suggested that the application of fumigation is preferable over that of temephos. It was also suggested that applying both control schemes simultaneously gives the most significant reduction in the population.
~\\
~\\
{\textbf{Keywords}}: \textsf{Mosquito population dynamics, mosquito reduction management, basic mosquito offspring number, optimal control}
\vspace*{-5pt}
\end{abstract}


\section{Introduction}
It is known that \textit{Aedes aegypti} mosquitoes be the primary transmitters of dengue fever in the world \cite{BBG02c}. These species were predominantly tropical species and confined to coastal areas. Now they are widespread inland and cause deadly morbidities by means of dengue fever in, mostly, Southeast Asia, Africa, USA, Australia, Brazil, Argentina, Caribbean Islands, China, India, Japan, and Portugal \cite{PT13c,EG07c}. The areas where \textit{Aedes aegypti} mosquitoes most probably nest, where they have also a tremendous threat of dengue outbreaks, are densely urbanized areas \cite{BHL03c,CS06c}. Several studies confirmed that controling mosquito population had constricted the area of dengue endemicity throughout the globe, beside it helped to bring down the cost regarding a high number of people under surveillance. Even vector control has been considerably inexpensive, a government's tendency to take out concern over continual attempt on it leads to a necessity of a more efficient and effective control management. In order to precisely manage the mosquito population, one needs to model the dynamics of such population and to further investigate the effectivity of control intervention that is acted to it.

\textit{Aedes aegypti} belongs to species that inhabit in domestic water containers: bath vessels, flowerpots, drums, tins, unused tyres, untreated swimming pools, or even in curved broads where it is possible for water to last on for long time \cite{LHG87c}. Both male and female adult mosquitoes feed on nectar. Only females require additional blood sources to obtain nutrients before producing eggs and during eggs' maturation. Note that a single female can lay down at one time about 100--200 immersing eggs in water \cite{CHF10c}. Potential eggs can be produced by \emph{Aedes aegypti} females up to five batches during their lifetime.  An egg needs 2--5 days to maturate its living embryo, depending on the water temperature \cite{JHH65c,MC11c}. In an advanced growth, each egg turns to a larva that posteriorly withstands for 5--10 days depending on the water temperature and air humidity \cite{GGM95c}. This living larva usually eats algae and microorganisms in the water surface, making each individual competes with its own for logistics. A discussion about competition among larvae becomes important after the fact that in a joint container, \textit{Aedes albopictus} larvae outcompete \textit{Aedes aegypti} larvae, thus the winners develop at a faster rate \cite{Bar96c}. In the prescribed range of lifetime, each larva undergoes four times skin exfoliation and turns to a pupa at the very end of the processes. An idle pupa needs to wait for 1--5 days before it metamorphoses into an adult \cite{MRB12c}. In addition, the living adult can generally survive for 10 days, or in some extreme cases for 2--4 weeks \cite{ZK10c}.

\textsc{Tabachnick} \emph{et al.} \cite{TMP78c} found three polytypic origins where \textit{Aedes aegypti} breeds: domestic (urban housing including its narrowing environment), sylvan (rural areas, and for some cases, tree holes and leaf axils), and peridomestic (artificial plantation areas such as coconut groves and farms). In general, separation of the origins as indoor and outdoor had also been highlighted. In \cite{IAA12c}, the authors notified that the number of \textit{Aedes aegypti} differed based on indoor-outdoor classification and heterogeneity of containers. A brief corresponding result showed that \textit{Aedes aegypti} species constituted as the most abundant species in indoor containers as compared to the other tested species. Some reference also mentioned that in indoor containers, the competition among larvae had not always been the case since a particular tendency made \textit{Aedes aegypti} grew more than \textit{Aedes albopictus} \cite{CHC71c}. 

By a basic idea of incorporating control measures to a mathematical model of mosquito population dynamics, several control schemes were tested towards fighting the spread of \textit{Aedes aegypti} mosquitoes. We highlight, for instances, the utilization of ultra low volume (ULV) insecticide \cite{LCM12c} combined with temephos \cite{WGS13c}, also the sterile insect technique (SIT) \cite{CT75c,PM87c,TYE10c,FMO13c} as a Genetic-based Vector Control (GVC). However, in this paper we do not take the epidemiology within mosquito population, i.e. the infection exposures, into account. Therefore, a susceptible-infected segregation is no longer in use. We accentuate a model of mosquito population dynamics based on indoor-outdoor life cycle classification and an introduction of control intervention as well as the investigation over the best strategy for reducing the population with cost as cheap as possible. As control apparatuses, we consider temephos (mainly to kill larvae) and fumigation (to kill adults by blocking their respiration process) after the fact that they have already been well-used apparatuses in an integrated mosquito eradication programme.

We organize the rest of the paper as follows. In Section 2, we set up a mathematical model capturing the dynamics of age-segregated mosquito population based on indoor-outdoor life cycle classification. We add two control measures to the system in order to check how the population trajectories respond to such control. In Section 3, we examine the biological meaningfulness of the model based on the positivity of the system's solution, also the existence and stability of equilibria. So far we use plausible constant values for the control measures. In Section 4, we discuss an optimal control model that generates time-variant control measures solving a proposed optimization problem. In seeking an optimal solution, we first use the indirect method to generate the state-adjoint-gradient system and then use the gradient method to solve the generated system in an algorithmic workflow. In Section 5, we do some numerical tests to bring forward visualizations of the model.

\section{The mathematical model}
To capture the dynamics of mosquito population, let us first classify the population into compartments based on age-segregation: egg $E$, larva $L$ and adult $A$. We are not concerned with the dynamics of pupae based on our assumption that there will be no deaths and no imbalance inflows and outflows contributed in the changes of the population of pupae, since then the population remains constant. Based on origin classification, the eggs and larvae differ from indoor $E_1,L_1$ and outdoor $E_2,L_2$. No more segregation for the adults since every individual can fly to wherever it prefers, indoor as well as outdoor. We introduce two control measures $u_1$ and $u_2$ representing the rates of use of temephos and fumigation, respectively. Let $[0,T]$ be the range of observation time $t$ and $x=(E_1,E_2,L_1,L_2,A)\in L^2([0,T];\R^5)$ be the state variable. Our model is preliminarily exhibited by the following equation
\begin{equation}\label{eq:stateaa}
\dot{x}(t)=f(x(t),u(t)),\quad x(0)=x_0\succeq 0.
\end{equation}
We assume that adult mosquitoes select indoor breeding sites with the probability $p$ and, therefore, select outdoor breeding sites with the probability $1-p$. Alongside with the introduction of the tendency-based probability, we denote $\mu$ as the corresponding rate of the adults to lay eggs. It is assumed that natural deaths can occur in all compartments, thus we denote $\eta_{\{1,\cdots,5\}}$ as the corresponding rates. In an average period, both indoor and outdoor eggs metamorphose into larvae with the transition rates $\alpha_{\{1,2\}}$. The same situations hold for indoor and outdoor larvae respectively, that they metamorphose into adults with the transition rates $\beta_{\{1,2\}}$. As an underlying discussion, we introduce two logistic coefficients $\sigma_{\{1,2\}}$ accounting for the phenomena of competition among larvae. Assume that all indoor or outdoor water containers are treated as homogeneous such that $\frac{\sigma_{1}}{\sigma_{2}}$ proportionates to some constant $M$. A suitable choice for $M>1$ leads to the situation where outdoor breading sites have $M$-times larger carrying capacity than indoor breading sites. After all, the model \eqref{eq:stateaa} is unfolded as
\begin{subequations}
    \label{eq:compwiseaa}
\begin{align}
    \dot{E}_1 &= \mu p A - (\alpha_1+\eta_1) E_1 - q u_1 E_1\\
    \dot{E}_2 &= \mu (1-p) A - (\alpha_2+\eta_2) E_2 \\
    \dot{L}_1 &= \alpha_1 E_1 - \sigma_1 L_1^2 - (\beta_1+\eta_3) L_1 - u_1 L_1 \\
    \dot{L}_2 &= \alpha_2 E_2 - \sigma_2 L_2^2 - (\beta_2+\eta_4) L_2\\
    \dot{A} &= \beta_1 L_1 + \beta_2 L_2 - \eta_5 A - u_2 A.
\end{align}
\end{subequations}
Two control measures are added into the model: the rate of use of temephos $u_1$ and that of fumigation $u_2$. Temephos is disseminated into indoor water containers to kill larvae and eggs. The weighting factor $q\in \left[0,\tfrac12\right]$ accounts for the fact that the dissemination has less impact on eggs. Meanwhile, fumigation directly targets adult mosquitoes.

It is assumed that all the controls $u$ belong to a set of admissible control $U\subset L^2([0,T];\R^2_+)$. Define 
\begin{equation}
X:=\set{x:\dot{x}(t)=f(x(t),u(t)),t\in[0,T],x(0)=x_0\succeq 0,u\in U}
\end{equation}
as the set of feasible states. It is intuitively believed that the higher $u$ imposed to the model, the smaller the population size. On the other hand, the smaller $u$, the larger the population size. The larger $u$ means that a policy maker needs to spend more funds, otherwise there will be no significant reduction to the mosquito population size. Given positive trade-off constants $\omega_{x,\{1,\cdots,5\}}$ and $\omega_{u,\{1,2\}}$, the following objective functional accommodates the necessity of balancing the situation between significant population reduction and limitation of funds:
\begin{equation}
\label{eq:obj}
   J(u) :=\frac1{2 T} \int_0^T \sum_{i=1}^5 \omega_{x,i} x_i^2(t)
   	+\sum_{j=1}^2 \omega _{u,j} u_j^2(t)\, \dd t.
\end{equation}
Now our optimization problem reads as
\begin{equation}
\text{find }(x,u)\in X\times U\text{ such that }J(u)\rightarrow\min.
\end{equation}

\section{Model analysis}
We need to ensure that our model is biologically meaningful. The following theorem gives a primary meaningfulness of the model: whenever the initial condition is positive, then the solution points in the forward time stays positive.
\begin{theorem}
Consider the model \eqref{eq:stateaa} where $u\in U$. If $x_0\succeq0$ then $x(t)\succeq0$ for all $t>0$.
\end{theorem}
\begin{proof}
Following the steps in \cite{WGS13c}, we let $n$ be a $(5\times5)$-matrix representing a collection of all normal vectors (by rows) to the boundary of nonnegative orthant $\partial\R^5_+$. Thus we have $n=-I_5$ where $I_5$ denotes the identity matrix. To ensure that the solution trajectory does not walk out of the nonnegative orthant, one only needs to check the solution points in the boundary. Notice that at $i$-th boundary, $\partial_i\R^5_+$, 
\begin{equation*}
\left.\left[n f(x,u)\right]_i\right|_{x\in\partial_i\R^5_+,u\in U}\leq 0.
\end{equation*}
This means that the direction of the evolution of the a solution point in a boundary is in counter-direction or at least perpendicular to the corresponding normal vector. Thus, it follows that the solution must not leave $\R^5_+$ for all $t>0$.
\end{proof}
Let us first consider the autonomous system~\eqref{eq:stateaa} with constant controls $u\in\R^2_+$. For abbreviation purpose, let $s_1=\mu p$, $s_2=\mu(1-p)$, $s_3=\alpha_1$, $s_4=\alpha_2$, $s_5=\beta_1$, $s_6=\beta_2$ and $d_1=\alpha_1+\eta_1+qu_1$, $d_2=\alpha_2+\eta_2$, $d_3=\beta_1+\eta_3+u_1$, $d_4=\beta_2+\eta_4$, $d_5=\eta_5+u_2$. In order to obtain equilibria of the system, we need to solve $f(x,u)=0$. It follows from this process that
\begin{subequations}
\label{eq:E}
\begin{eqnarray}
A^{\ast}&=&\frac{s_5L_1^{\ast}+s_6L_2^{\ast}}{d_5}\\
E_1^{\ast}&=&\frac{s_1 (s_5L_1^{\ast}+s_6L_2^{\ast})}{d_1d_5}\\
E_2^{\ast}&=&\frac{s_2 (s_5L_1^{\ast}+s_6L_2^{\ast})}{d_2d_5}
\end{eqnarray}
where $(L_1^{\ast},L_2^{\ast})$ satisfies the leading equations
\begin{eqnarray}
L_1^2+a_1L_1+b_1L_2&=&0\label{eq:L1}\\
L_2^2+a_2L_2+b_2L_1&=&0.\label{eq:L2}
\end{eqnarray}
\end{subequations}
All the constants that belong to the last equations are given by
\begin{equation*}
a_1=\frac{d_3}{\sigma_1}-\frac{s_1s_3s_5}{d_1d_5\sigma_1},\,
b_1=-\frac{s_1s_3s_6}{d_1d_5\sigma_1},\,
a_2=\frac{d_4}{\sigma_2}-\frac{s_2s_4s_6}{d_2d_5\sigma_2},\,
b_2=-\frac{s_2s_4s_5}{d_2d_5\sigma_2}.
\end{equation*}
It yields from \eqref{eq:L1}--\eqref{eq:L2} three equilibria
\begin{eqnarray*}
Q_1&=&(0,0,0,0,0)\\
Q_2&=&\left(-\frac{s_1s_5a_1}{d_1d_5},-\frac{s_2s_5a_1}{d_2d_5},-a_1,0,-\frac{s_5a_1}{d_5}\right)\\
Q_3&=&\left(-\frac{s_1s_6a_2}{d_1d_5},-\frac{s_2s_6a_2}{d_2d_5},0,-a_2,-\frac{s_6a_2}{d_5}\right).
\end{eqnarray*}
The definitions of $Q_2$ and $Q_3$ only make sense in biology if only $a_1,a_2$ are negative. As if they are stable, it simply means that after very long time, either the compartment of indoor larvae or that of outdoor larvae tends to extinction while the other compartments stay alive. In the biological context, such situation can hardly happen. In the next writing, we demonstrate that the choices of $a_1,a_2$ being positive lead to a more interesting discussion. First, we introduce a measure whose cubic value is given by
\begin{equation}\label{eq:Raa}
\RM(u)^3:=\frac{s_1s_3s_5}{d_5d_1d_3}+\frac{s_2s_4s_6}{d_5d_2d_4}.
\end{equation}
The following identity is found after some algebraic manipulations
\begin{equation}\label{eq:id}
\frac{b_1b_2-a_1a_2}{d_3d_4}\sigma_1\sigma_2=\RM(u)^3-1.
\end{equation}
Using one as the threshold value, the following theorem justifies the stability of zero equilibrium $Q_1$ by considering the nominal of $\RM(u)$ relative to the threshold. Meanwhile, the next theorem shows the existence of a positive nontrivial equilibrium (also well-known as coexistence equilibrium).
\begin{theorem}
The zero equilibrium $Q_1$ is locally asymptotically stable if $\RM(u)<1$ and is unstable if $\RM(u)>1$.
\end{theorem} 
\begin{proof}
Checking the stability of the zero equilibrium $Q_1$ in a local view is similar to see the behavior of the solution of the linearized version of \eqref{eq:stateaa} around $Q_1$, i.e. $\dot{x}=\nabla f(Q_1,u)(x-Q_1)$. In this case, $\nabla f(Q_1,u)$ stands for the \textsc{Jacobi}an of $f$ evaluated at $Q_1$. Unfolding this $\nabla f(Q_1,u)$, we get
\begin{equation*}
\nabla f(Q_1,u)=\left[ \begin {array}{ccccc} 
-d_1&0&0&0&s_1\\
0&-d_2&0&0&s_2\\
s_3&0&-d_3&0&0\\
0&s_4&0&-d_4&0\\
0&0&s_5&s_6&-d_5
\end {array} \right]. 
\end{equation*}
Here, $Q_1$ is locally asymptotically stable if and only if all eigenvalues of $\nabla f(Q_1,u)$ have negative real part. To see this, a simple cofactorization method computes the determinant of $\nabla f(Q_1,u)-\lambda I_5$ as the sum of all its cofactors, given by the negative of 
\begin{equation*}
(d_1+\lambda)(d_2+\lambda)(d_3+\lambda)(d_4+\lambda)(d_5+\lambda)-s_2s_4s_6(d_1+\lambda)(d_3+\lambda)-s_1s_3s_5(d_2+\lambda)(d_4+\lambda).
\end{equation*}
Clearly, $\text{coeff}\left(\lambda^{\{3,4,5\}}\right)$ are real positive, meanwhile
\begin{eqnarray*}
\text{coeff}(\lambda^0)&=&d_1d_2d_3d_4d_5-s_1s_3s_5d_2d_4-s_2s_4s_6d_1d_3\\
\text{coeff}(\lambda^1)&=&d_1d_2d_3d_4+\cdots+d_2d_3d_4d_5-s_1s_3s_5(d_2+d_4)-s_2s_4s_6(d_1+d_3)\\
\text{coeff}(\lambda^2)&=&d_1d_2d_3+\cdots+d_3d_4d_5-s_1s_3s_5-s_2s_4s_6.
\end{eqnarray*}
If $\RM(u)<1$, or equivalently $\frac{s_1s_3s_5}{d_5d_1d_3}+\frac{s_2s_4s_6}{d_5d_2d_4}<1$, one can see that $\text{coeff}\left(\lambda^{\{0,1,2,3,4,5\}}\right)$ are real positive. Since $\lim_{\lambda\rightarrow\pm\infty}-\det(\nabla f(Q_1,u)-\lambda I_5)=\pm\infty$, $\left.-\det(\nabla f(Q_1,u)-\lambda I_5)\right|_{\lambda=0}>0$ and
\begin{equation*}
\left.-\frac{\dd \det(\nabla f(Q_1,u)-\lambda I_5)}{\dd \lambda}\right|_{\lambda>0}>0,\quad \left.-\frac{\dd^2 \det(\nabla f(Q_1,u)-\lambda I_5)}{\dd \lambda^2}\right|_{\lambda>0}>0,
\end{equation*}
then $\det(\nabla f(Q_1,u)-\lambda I_5)$ cannot have zero that has positive real part. Moreover, all zeros of $\det(\nabla f(Q_1,u)-\lambda I_5)$ have negative real part. If $\RM(u)>1$, then $\text{coeff}(\lambda^0)<0$. Then for however the values of $\text{coeff}(\lambda^1)$ and $\text{coeff}(\lambda^2)$, it follows that there must be at least one real positive zero. This means $Q_1$ is unstable. If $\RM(u)=1$, then a similar step can prove that $\det(\nabla f(Q_1,u)-\lambda I_5)$ has one zero with the value $0$ and the other zeros remain with negative real part. In this case, $Q_1$  is stable if the origin in the reduced system in the center manifold is stable.
\end{proof}

\begin{theorem}
If $a_1$ and $a_2$ are positive, a unique coexistence equilibrium $Q_{4}$ exists if $\RM(u)>1$ and does not if $\RM(u)\leq 1$. If $a_1$ and $a_2$ are negative, then there always exists a coexistence equilibrium for any value of $\RM(u)$. Else, then it always holds $\mathcal{R}(u)>1$, and therefore, there always exists a coexistence equilibrium.
\end{theorem} 
\begin{proof}
A simple substitution in \eqref{eq:L1}-\eqref{eq:L2} makes the equilibrium state $L_{1}$ following
\begin{equation*}
L_{1}^{4}+2a_1L_{1}^{3}+(a_1^{2}-b_1a_2)L_{1}^{2}+(b_1^{2}b_2-a_1b_1a_2)L_{1}=0.
\end{equation*}
As $L_{1}=0$ results in equilibria that have been discussed before, now let us consider the remaining cubic
\begin{equation*}
P(L_{1}):=L_{1}^{3}+2a_1L_{1}^{2}+(a_1^{2}-b_1a_2)L_{1}+(b_1^{2}b_2-a_1b_1a_2).
\end{equation*}
Since the coefficient of $L_{1}^{3}$ is $1>0$, then $\lim_{L_{1}\rightarrow \pm \infty}P=\pm \infty$. We prove this theorem by considering the following cases.

\textbf{Case 1}: $a_1$ and $a_2$ are positive. In this case, three possibilities occur: $\RM(u)>1$, $\RM(u)<1$ and $\RM(u)=1$. If $\RM(u)>1$, then by the identity \eqref{eq:id} it holds that $b_1b_2>a_1a_2$. Thus we have $0>b_1^{2}b_2-a_1b_1a_2=P(0)$ since $b_1<0$. It follows from the value of $P(0)$ and $\lim_{L_{1}\rightarrow \infty}P= \infty$ that $P$ must have at least one positive root. Now we have to check that $L_{2}=\frac{L_{1}(L_{1}+a_1)}{-b_1}>0$. To see this, let $L_{1}=\rho^{2}$ be the corresponding root, where $\rho\in\R\backslash\{0\}$. Extracting $(L_{1}-\rho^{2})$ out of the cubic, we get $L_{1}^{2}+(2a_1+\rho^2)L_{1}+\mathcal{N}$ where
\begin{equation*}
\mathcal{N}=(a_1^{2}-b_1a_2)+\rho^{2}(2a_1+\rho^2).
\end{equation*}
Since $\mathcal{N}=\frac{-(b_1^{2}b_2-a_1b_1a_2)}{\rho^2}>0$ then together with the claim of the uniqueness of the positive root $L_{1}=\rho^{2}$, we need $(2a_1+\rho^2)>0$ to ensure that the remaining quadratic function does not have a real positive root. Since $a_1>0$, then $(\rho^2+a_1)>0$. Now we have proved that $L_{2}$ is positive, and by \eqref{eq:E}, the other states are also positive. If $\RM(u)<1$, then $P(0)>0$. Since $\lim_{L_{1}\rightarrow -\infty}P= -\infty$, this means there exists at least one negative root $L_{1}=-\varepsilon^2$ of the cubic, where $\varepsilon\in\R\backslash\{0\}$. Now extract $(L_{1}+\varepsilon^2)$ out of the cubic and we get the remaining quadratic function $L_{1}^{2}+(2a_1-\varepsilon^2)L_{1}+\mathcal{M}$ where
\begin{equation*}
\mathcal{M}:=(a_1^2-b_1a_2)-\varepsilon^2(2a_1-\varepsilon^2).
\end{equation*}
It turns out that $\mathcal{M}=\frac{b_1^{2}b_2-a_1b_1a_2}{\varepsilon^{2}}=\frac{P(0)}{\varepsilon^{2}}>0$. Thus further analysis confirms that a positive root exists whenever $(2a_1-\varepsilon^2)$ is strictly less than zero and the discriminant $\Delta=(2a_1-\varepsilon^2)^{2}-4\mathcal{M}\geq 0$. One can see that $\Delta=(2a_1-\varepsilon^2)^{2}-4(a_1^2-b_1a_2)+4\varepsilon^2(2a_1-\varepsilon^2)=(2a_1+3\varepsilon^2)(2a_1-\varepsilon^2)-4(a_1^2-b_1a_2)<0$ since $a_1,a_2>0$ and $b_1<0$. This contradicts the necessity that $\Delta\geq 0$. Then there does not exist any positive root of the cubic. At last, if $\RM(u)=1$, then $b_1^{2}b_2-a_1b_1a_2=0$ and we get the remaining quadratic function $L_{1}^{2}+2a_1L_{1}+(a_1^{2}-b_1a_2)$. Since $a_1,a_2>0$ and $b_1<0$, it is clear that the quadratic function does not have any real positive root.

\textbf{Case 2}: $a_1$ and $a_2$ are negative. If $\RM(u)>1$, then it is easily seen that one positive root exists. Now we have to check that $L_{2}=\frac{L_{1}(L_{1}+a_1)}{-b_1}>0$. To see this, as it is previously done together with the claim of the uniqueness of the positive root $L_{1}=\rho^{2}$, we need $(2a_1+\rho^2)>0$ to ensure that the remaining quadratic function does not have any real positive root. Then, $(\rho^2+a_1)>(2a_1+\rho^2)>0$ since $a_1<0$. Now we have proved that $L_{2}$ is positive, and therefore, the other states are also positive. If $\RM(u)=1$, we have $b_1^{2}b_2-a_1b_1a_2=0$ and the remaining quadratic function is $L_{1}^{2}+2a_1L_{1}+(a_1^{2}-b_1a_2)$. Since $\Delta=4a_1^{2}-4(a_1^{2}-b_1a_2)=4b_1a_2>0$ and $2a_1<0$, then a positive root exists whose uniqueness is confirmed by the value of $(a_1^{2}-b_1a_2)$. Thus, we need to proof that $L_{2}=\frac{L_{1}(L_{1}+a_1)}{-b_1}>0$. Let $L_{1}=\eta^{2}$ be such a root, then the other root must be $L_{1}=-(2a_1+\eta^{2})$ where $(2a_1+\eta^{2})=-\frac{(a_1^{2}-b_1a_2)}{\eta^{2}}$. However, the positive root is unique if $(a_1^{2}-b_1a_2)\leq 0$, leading to $\eta^{2}+a_1>\eta^{2}+2a_1=-\frac{(a_1^{2}-b_1a_2)}{\eta^{2}}\geq 0$. This confirms that $L_{2}>0$. If $\RM(u)<1$, then as before, we have at least one negative root $L_{1}=-\varepsilon^2$ and the remaining quadratic $L_{1}^{2}+(2a_1-\varepsilon^2)L_{1}+\mathcal{M}$. Observe that $\mathcal{M}>0$ and $(2a_1-\varepsilon^2)<0$ since $a_1<0$. Some factorization confirms that there exist two remaining positive roots $L^{(1,2)}_{1}=\{\rho^{2},\xi^{2}\}$ where $\xi^{2}=-(2a_1-\varepsilon^2+\rho^{2})$ and $\rho^{2}\xi^{2}=\mathcal{M}$. Then, we have to check the result of $L^{(1,2)}_{1}+a_1$ since this confirms the positivity of $L_{2}$, or furthermore, existence of a positive equilibrium. Since $(\rho^{2}+a_1)+(\xi^{2}+a_1)=-(2a_1-\varepsilon^2)+2a_1=\varepsilon^2>0$, this ensures that at least one of the roots is greater than $-a_1$, since then at least one root makes $L_{2}$ positive.

\textbf{Case 3}: else. Since it always holds $b_1b_2>a_1a_2$, this implies $\RM(u)>1$ and $P(0)=b_1^2 b_2-a_1b_1a_2<0$. Therefore, there is always (not necessarily unique) coexistence equilibrium. 
\end{proof}

All compartments have responsibility to generate offsprings in the sense that newborns are impossible whenever one compartment in the system remains zero for all time. Regarding this, we define a mosquito-free equilibrium (MFE) as a static condition where all compartments cannot reproduce for the next offsprings. This means that MFE is equivalent to the zero equilibrium $Q_1$. Following the next generation method \cite{DW02c}, we 
define two matrices $F(u),V(u)$ where $F(u)$ is a matrix with zeros in the main diagonal and $V(u)$ is a positive diagonal matrix such that $\nabla f(Q_1,u)=F(u)-V(u)$. Then we define a so-called \textit{next generation matrix} $G(u)=F(u)V(u)^{-1}$ which is nothing else but 
\begin{equation}
G(u) = \left(
\begin{array}{ccccc}
0 & 0 & 0 & 0 & \frac{s_1}{d_5}\\
0 & 0 & 0 & 0 & \frac{s_2}{d_5}\\
\frac{s_3}{d_1} & 0 & 0 & 0 & 0\\
0 & \frac{s_4}{d_2} & 0 & 0 & 0 \\
0 & 0 & \frac{s_5}{d_3} & \frac{s_6}{d_4} & 0
\end{array}
\right).
\end{equation}
The $(i,j)$-th element of this matrix represents the
average number of new individuals in the compartment $i$ produced by a single individual
from the compartment $j$ during the compartment $j$'s average individual lifetime period.
It is easy to verify that the spectral radius of $G(u)$ is given by
\begin{equation}
\max\set{|\lambda|:u\in\mathbb{R}^2_+,\,\det(G(u)-\lambda I_{5})=0}=\RM(u).
\end{equation}
This function $\RM=\RM(u)$ is often called as \textit{the basic mosquito offspring number}. The formulation in \eqref{eq:Raa} provides the dependency of the basic mosquito offspring number on several parameters and the control $u$. In Section~\ref{sec:numerbb}, we check the behaviour of the basic mosquito offspring number with respect to the control and several numerically unspecified parameters.

\begin{theorem}\label{thm:coexbb}
The existing coexistence equilibrium $Q_4$ is globally asymptotically stable in the nonnegative orthant $\R^5_+$ if $\RM(u)>1$.
\end{theorem}

\begin{remark}
We omit writing the details of the proof of Theorem~\ref{thm:coexbb}. Furthermore, it requires the use of a \textsc{Lyapunov} function and the \textsc{LaSalle}'s Invariance Principle. As another reference, one can further see from numerical results that the solution of the model tends to $Q_4$ asymptotically, under a supplemented condition that the set of parameters is chosen such that $\RM(u)>1$.  
\end{remark}

\section{Optimal control problem}

Recall our optimization problem: 
\begin{equation}
\label{eq:optaa}
\text{ find }(x,u)\in X\times U\text{ such that }J(u)\rightarrow\min.
\end{equation}
Denote by $U$ a set of admissible controls where it is assumed to be compact. Let $g$ be a function such that $J(u)=\int_{0}^{T}g(x,u)\,\dd t$. Let $\bar{u}$ be an optimal control that solves \eqref{eq:optaa}. Consider a small variation around the optimal control
\begin{equation}
u^{\epsilon}(t)=\bar{u}(t)+\epsilon\kappa(t),\quad \abs{\epsilon}\ll 1.
\end{equation}
Plugging this variation into \eqref{eq:stateaa} together with a nonnegative initial condition, we have the resulting perturbed state $x^{\epsilon}$, where it holds
\begin{equation}
\dot{x}^{\epsilon}-\dot{\bar{x}}=f(x^{\epsilon},u^{\epsilon})-f(\bar{x},\bar{u}).
\end{equation}
Note from the model that $f\in C^{\infty}(\R^5\times\R^2;\R^5)$, thus we have the following Taylor expansion
\begin{equation}
\frac{\dd}{\dd t}(x^{\epsilon}-\bar{x})=f_{x}(\bar{x},\bar{u})(x^{\epsilon}-\bar{x})+\epsilon f_{u}(\bar{x},\bar{u})\kappa+\mathcal{O}(\epsilon^2).
\end{equation}
Working at $\mathcal{O}(\epsilon^2)$, we let $\varphi$ be the solution of the differential equation
\begin{equation}\label{eq:varphiaa}
\dot{\varphi}=f_{x}(\bar{x},\bar{u})\varphi+f_{u}(\bar{x},\bar{u})\kappa,\quad \varphi(0)=0.
\end{equation}
Thus, some algebraic computations show that our perturbed state is given by
\begin{equation}\label{eq:xeaa}
x^{\epsilon}=\bar{x}+\epsilon\varphi+\mathcal{O}(\epsilon^2).
\end{equation}
Further it can be proven that the solution of \eqref{eq:varphiaa} exists and therefore the term $\frac{\partial x^{\epsilon}}{\partial\epsilon}$ exists.

Let $z\in L^2([0,T];\R^5)$ be some dual variable. It follows that $\int_{0}^{T}\frac{\dd}{\dd t}(z\cdot x^{\epsilon})\,\dd t=z(T)\cdot x^{\epsilon}(T)-z(0)\cdot x_{0}$. 
Append this to the objective functional and we get the following expression
\begin{equation}
J(u^{\epsilon})=\int_{0}^{T}\left[g(x^{\epsilon},u^{\epsilon})+\frac{\dd}{\dd t}(z\cdot x^{\epsilon})\right]\,\dd t+ z(0)\cdot x_{0} - z(T)\cdot x^{\epsilon}(T).\label{eq:Joptaa}
\end{equation}
Since $\bar{u}$ is a local minimizer of $J$, then derivative of $J$ over $\epsilon$ where $\epsilon=0$ exists and equals to zero, in other words $\left.\frac{\partial J(u^{\epsilon})}{\partial \epsilon}\right|_{\epsilon=0}=0$. Working a bit more on \eqref{eq:Joptaa} together with simplifying the result by factorizations, we get
\begin{equation}\label{eq:allaa}
0=\int_{0}^{T}(g_{x}+f_{x}'z+\dot{z})\cdot\left.\frac{\partial x^{\epsilon}}{\partial \epsilon}\right|_{\epsilon=0}\,\dd t+\int_{0}^{T}(g_{u}+f_{u}'z)\cdot\kappa\,\dd t-z(T)\cdot\left.\frac{\partial x^{\epsilon}}{\partial \epsilon}\right|_{\epsilon=0}(T).
\end{equation}
Note from \eqref{eq:varphiaa}-\eqref{eq:xeaa} that a variation of $\kappa$ makes $\frac{\partial x^{\epsilon}}{\partial \epsilon}$ varying. Zeroing the right-hand side of \eqref{eq:allaa} together with taking $\kappa=g_{u}+f_{u}'z$, we get the following summary.

\begin{theorem}
Consider the optimization problem \eqref{eq:optaa}. Let $\bar{u}\in U$ be a local minimizer of $J$ and $\bar{x}\in X$ be the resulting state. Then there exists a dual variable $\bar{z}\in L^2([0,T];\R^5)$ such that the tuple $(\bar{x},\bar{u},\bar{z})$ satisfies the following system
\begin{equation}\label{eq:neccaa}
\dot{x}=\frac{\partial\mathcal{H}}{\partial z}\text{ with }x(0)=x_0\succeq0,\quad 
\dot{z}=-\frac{\partial\mathcal{H}}{\partial x},\quad
\frac{\partial\mathcal{H}}{\partial u}=0, \quad z(T)=0
\end{equation}
for all $t\in[0,T]$. The function $\mathcal{H}(x,u,z):=g(x,u)+z'f(x,u)$ is the Hamiltonian function, meanwhile all equations in \eqref{eq:neccaa} are respectively the \textit{state}-\textit{adjoint}-\textit{gradient} system and the \textit{transversality condition}.
\end{theorem}
The adjoint (with the transversality condition) and gradient equations can now be unfolded as
\begin{subequations}
\label{eq:adjaa}
\begin{align}
\dot{z}_{1} &=-\omega_{x,1}E_1+ (\alpha_{1}+ qu_{1}+ \eta_{1})z_{1}- \alpha_{1}z_{3},&z_1(T)=0\\
\dot{z}_{2} &= -\omega_{x,2}E_2+ (\alpha_{2}+
\eta_{2})z_{2}-\alpha_{2}z_{4},&z_2(T)=0\\
\dot{z}_{3} &=-\omega_{x,3}L_{1}+ (2\sigma_{1}L_{1} +\beta_{1}+ u_{1}+\eta_{3})z_{3}- \beta_{1}z_{5},&z_3(T)=0\\
\dot{z}_{4} &=-\omega_{x,4}L_{2} +(2\sigma_{2}L_{2} +
\beta_{2}+ \eta_{4})z_{4}- \beta_{2}z_{5},&z_4(T)=0\\
\dot{z}_{5}&=-\omega_{x,5}A+(u_{2}+ \eta_{5})z_{5}- \mu p z_{1} - \mu(1-p) z_{2},&z_5(T)=0
\end{align}
\end{subequations}
and
\begin{subequations}
\label{eq:gradaa}
\begin{eqnarray}
\omega_{u,1}u_{1}- qE_{1}z_{1}-L_{3}z_{3}&=&0\\
\omega_{u,2}u_{2} -Az_{5}&=&0.
\end{eqnarray}
\end{subequations}
It is essential that $U$ is bounded as this supports the meaningfulness of the control from the application point of view. For the sake of simplicity, we consider $U=L^2([0,T],[v_1,w_1]\times[v_2,w_2])$ and the according projection
\begin{equation}\label{eq:sataa}
\text{\textbf{sat}}(u)=\max\left(v,\min\left(w,u\right)\right)
\end{equation}
mapping any control $u$ into $U$. 

To solve the optimal control problem~\eqref{eq:optaa} for a bounded control $u \in U$, we first consider the unconstrained convex optimization problem
\begin{equation*}
\text{ find }(x,u)\in L^2(0,T)\times L^2(0,T)\text{ such that }J(u)\rightarrow\min
\end{equation*}
and use the projection~\eqref{eq:sataa} to map the control into a convex set $U$. Since $U$ is a constriction of $L^2(0,T)$, the projection leads us to the optimal solution of~\eqref{eq:optaa}. After all, \ref{alg:graddbb} illustrates our scheme to solve \eqref{eq:optaa}.
\begin{mdframed}
\begin{description}
\item[\setword{\textsc{Algorithm 1}}{alg:graddbb}]
\item[Return:] The tuple $(\hat{x},\hat{u},\hat{J})$.
\item[Step 0] Set $k=0$, an initial guess for the control $u^{k}\in U $, an error tolerance $\epsilon>0$ and an initial step length $\lambda>0$.
\item[Step 1] Compute $x^k_{u}\leftarrow x^{k}(\cdot;u^k)$ and $z^k_{u}\leftarrow z^{k}(\cdot;x^{k}(\cdot;u^k))$ consecutively from the state (with a forward scheme) and adjoint equation (with a backward scheme).
\item[Step 2] Compute the objective functional $J(u^k)$.
\item[Step 3] Compute $\bar{u}^k(x^{k}_{u},z^{k}_{u})$ from the gradient equation.
\item[Step 4] Update $u^{k+1}(\lambda)\leftarrow u^k+\lambda \bar{u}^k$ and $u^{k+1}\leftarrow\text{\textbf{sat}}(u^{k+1})$. Compute $x^{k+1}_{u}$ and $z^{k+1}_{u}$.
\item[Step 5] Compute $J(u^{k+1})$ and set $\Delta J\leftarrow J(u^{k+1})-J(u^{k})$.
\item[Step 6] If $|\Delta J|<\epsilon$, then set $(\hat{x},\hat{u},\hat{J})\leftarrow(x^{k+1}_{u},u^{k+1},J(u^{k+1}))$ and stop.
\item[Step 7] While $\Delta J\geq 0$ do
\begin{enumerate}[label=(7.\arabic*)]
\item Update new $\lambda\leftarrow\arg\min_{s\in[0,\lambda]}\phi(s):=J(u^{k+1}(s))$ where $\phi$ is a quadratic representation of $J$ with respect to the step length $s$. Note that the solution exists since $\phi(0)$, $\phi'(0)$ and $\phi(\lambda)$ can be computed directly. 
\item Compute to the new $u^{k+1}(\lambda)\leftarrow u^k+\lambda \bar{u}^k$ and set $u^{k+1}\leftarrow\text{\textbf{sat}}(u^{k+1})$. Then compute the new $x^{k+1}_{u}$ and $z^{k+1}_{u}$.
\item Compute $J(u^{k+1})$ and set $\Delta J\leftarrow J(u^{k+1})-J(u^{k})$.
\item If $|\Delta J|<\epsilon$, then set $(\hat{x},\hat{u},\hat{J})\leftarrow(x^{k+1}_{u},u^{k+1},J(u^{k+1}))$ and stop.
\end{enumerate}
\item[Step 8] Set $k\leftarrow k+1$ and go to \textbf{Step 3}.
\end{description}
\end{mdframed}
\section{Numerical results}
\label{sec:numerbb}
Several discussions in this section aim at describing the impact of control intervention on the abundance of mosquito population in visual statements. We accentuate in this paper that significant reduction of the population in forward time can be a carry-over from only a cheap optimal control. Realistically, optimal control still suffers from practical drawbacks that one can not easily implement it in real situation. The reason is arisen by the picture of the control that is sometimes full of fluctuations. On the other hand, a constant control benefits from its easy-to-implement scheme as the executors only need to deal with some fixed-fund allotment problem with a flat monthly distribution. Another benefit from using a constant control is that the magnitudes of the basic mosquito offspring number can be traced. This trace deduces some important statements regarding the endemicity of the observed area. It has been preliminarily known that whenever the basic mosquito offspring number is less than one, then two facts arise: the zero equilibrium is stable but a coexistence equilibrium does not exist. If this number is greater than one, then the zero equilibrium is unstable and a coexistence equilibrium exists and is stable. The very last statement means that the endemicity will emerge and stay uninterruptible. However, a constant control always emerges with a higher cost as compared to an optimal control. 

In order to do some numerical tests, all the parameters involved in the model have to be represented in numbers. Table \ref{tab:1} gives the estimates of all parameters.
\begin{table*}[h!]
\begin{center}
\begin{tabular}{r|*{10}{l}}
\hline
Parameter &$M$ & $p$ & $\mu$ & $\eta_1$ & $\eta_2$ & $\eta_3$ & $\eta_4$ & $\eta_5$ & $q$ \\ 
Unit &- &- &$\text{day}^{-1}$ &$\text{day}^{-1}$ &$\text{day}^{-1}$ &$\text{day}^{-1}$ &$\text{day}^{-1}$ &$\text{day}^{-1}$ &- \\ 
Estimated value &$2$ & $0.4$ & $3.1$ & $0.02$ & $0.01$ & $0.002$ & $0.01$ & $0.4$ & $0.04$\\ \hline
\end{tabular}
\end{center}
\end{table*}
\begin{table*}[h!]
\begin{center}
\begin{tabular}{r|*{7}{l}}
\hline
Parameter &$\alpha_1$ & $\alpha_2$ & $\sigma_1$ & $\sigma_2$ & $\beta_1$ & $\beta_2$\\ 
Unit &$\text{day}^{-1}$ &$\text{day}^{-1}$ &$\text{ind.}^{-1}\times\text{day}^{-1}$ &$\text{ind.}^{-1}\times\text{day}^{-1}$ &$\text{day}^{-1}$ &$\text{day}^{-1}$ \\ 
Estimated value &$0.15$ & $0.13$ & $0.004$ & $\frac{\sigma_1}{M}$ & $0.08$ & $0.05$ \\ \hline
\end{tabular}
\end{center}
\end{table*}
\begin{table*}[h!]
\begin{center}
\begin{tabular}{r|*{7}{l}}
\hline
Parameter &$T$ & $\omega_{x,\{1,\cdots,4\}}$\hspace{0.4cm} & $\omega_{x,5}$\hspace{0.4cm} & $\omega_{u,\{1,2\}}$\hspace{0.4cm} & $v$ & $w$ \\ 
Unit&$\text{day}$\hspace{0.4cm} &- &- &- &$\text{day}^{-1}$\hspace{0.4cm} &$\text{day}^{-1}$ \\ 
Estimated value & $150$ &$1$ & $2$ & $4\times 10^4$ & $0$ & $[1,1]$\\ \hline
\end{tabular}
\caption{\label{tab:1}A set of values for all the parameters used in the model.}
\end{center}
\end{table*}
\subsection{Constant control}
To some reasons, the basic mosquito offspring number $\RM(u)$ reflects the hierarchy of endemicity of dengue. The higher $\RM(u)$, the more rapid the mosquito's growth is and therefore, the higher the number of dengue incidences will be. Here, we illustrate some impacts of using constant control on the magnitudes of $\RM(u)$. Alongside with this illustration, we show the dependence of $\RM(0)$ on some parameters whose numerical values may be hard to estimate. Such parameters can be like $p$ (the prevalence probability of mosquito adult to breed indoor) and $\mu$ (the birth rate of eggs). Figs.~\ref{fig:rubb} and~\ref{fig:rpbb} give the illustrations.
\begin{figure}[htbp!]
\begin{minipage}{.495\textwidth}
\includegraphics[width=\textwidth]{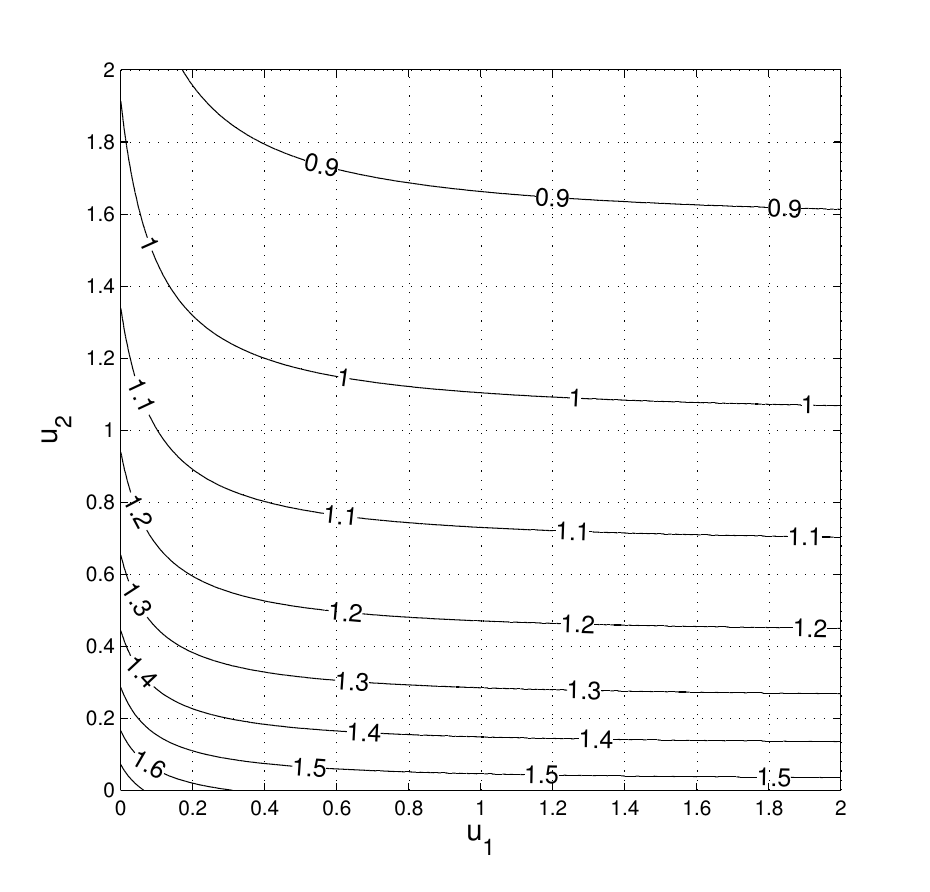}
\caption{\label{fig:rubb}Contour of the basic mosquito offspring number $\RM(u)$ in $(u_1,u_2)$-plane.}
\end{minipage}
\hfill
\begin{minipage}{.495\textwidth}
\includegraphics[width=\textwidth]{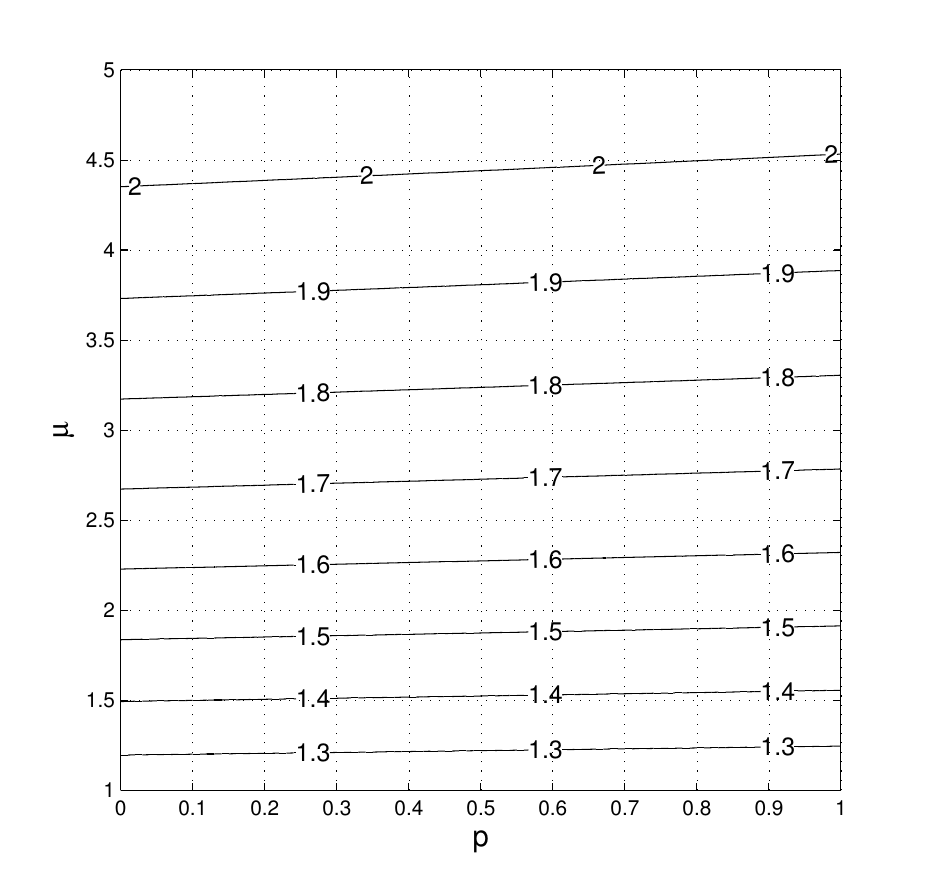}
\caption{\label{fig:rpbb}Contour of the basic mosquito offspring number $\RM(0)$ in $(p,\mu)$-plane.}
\end{minipage}
\end{figure}

From Fig.~\ref{fig:rubb}, some conclusions are drawn. It is clear from the figure that $u_2$ gives a more significant decrease to the magnitude of the basic mosquito offspring number as compared to what $u_1$ does. In the constant control case, each value of control taken in the set of real positive numbers can roughly represent a negative return. A negative return $u_1$ of $1$ simply means that $100\%$ the number of individuals in the previous time point (in day) has to be killed in the current time point by temephos dissemination. In practice, this requirement seems to be hard to achieve due to several technical and spatial heterogeneity problems. From Fig.~\ref{fig:rubb}, a suitable choice of the constant control pair in the range $[0,1]\times[0,1]$ still gives insignificant reduction to the number of mosquitoes. In an extreme case,  killing $100\%$ the number of individuals as in the previous time point using both temephos and fumigation still arises an endemic situation in the observed area. However, to achieve the condition $\RM(u)<1$, one has to produce a high value of negative return. For example, killing $40\%$ (of indoor larvae and $q$ times indoor eggs) and $180\%$ (of adults) as in the previous time point will produce $\RM(u)\approx 0.9$. Similarly, killing respectively $60\%$ and $120\%$ leads to $\RM(u)\lesssim 1$, meaning that there is a possibility for the mosquito population to be completely eradicated in the long run.

Meanwhile from Fig.~\ref{fig:rpbb}, for whatever the values of $(p,\mu)$ are taken in the range $[0,1]\times[1,5]$, the mosquitoes will never die out. The choice of all parameters in the model seems to be suitable with the real situation in an endemic region, that is, the basic mosquito offspring number with the absence of control must be greater than one. This is another reason why we need control intervention. In the figure, $\mu$ appears with a similar sensitivity as $u_2$ in Fig.~\ref{fig:rubb}. It can be seen that $\mu$ is very sensitive, as its slight changes lead to the significant differences in the hierarchy of the endemicity. In reality, $\mu$ can be a number that depends on meteorological parameters. It is believed that the more environmental condition sustain mosquitoes' life, the higher $\mu$ is. The result as in Fig.~\ref{fig:rp} can be another way to confirm that meteorology gives a significant impact on the mosquito abundance in an endemic region. The figure also tells us that, based on our model, to where adult mosquitoes prefer to breed is not really a matter. 

\subsection{Time-variant optimal control}
We consider the trajectory of the solution of the equation $\dot{x}=f(x,0),\,x_0\succeq0$, defined as $x(t;0)$. It can numerically be shown that $x(t;0)$ is monotonic and tends to the coexistence equilibrium $Q_4$ in the long run. Set $t_1=0$, $t_2=50$, $t_3=100$ and
\begin{equation}
X_i:=\{x:\dot{x}=f(x,u),\,t\in[0,T],x_0=x(t_i;0),u\in U \},\quad i=1,2,3.
\end{equation}
We divide our numerical test based on the following three scenarios:
\begin{align*}
&\text{find }(x,u)\in X_1\times U\text{ such that }J(u)\rightarrow\min,\tag{\textbf{Sc-1}}\label{sce:1bb}\\
&\text{find }(x,u)\in X_1\times U\text{ such that }J(u)\rightarrow\min,\tag{\textbf{Sc-2}}\label{sce:2bb}\\
&\text{find }(x,u)\in X_1\times U\text{ such that }J(u)\rightarrow\min.\tag{\textbf{Sc-3}}\label{sce:3bb}
\end{align*}%
Here we want to check the dynamics of the mosquito population in house-scale including its nearest neighborhood. On an average proportion, we let $x(t_1;0)=(8,8,6,6,5)$. Thus we obtained from our preliminary simulation with the absence of control that $x(t_2;0)=(200.41,354.84,73.82,134.85,30.56)$ and $x(t_3;0)=(283.12,514.21,91.19,168.22,39.16)$, respectively. 

One purpose of dividing the initial condition is to check which scenario arises with the cheapest cost, the one starts from $t=t_1$ (the earliest growth time), $t=t_2$ (the peak of outbreak) or $t=t_3$ (the population almost reach the coexistence equilibrium). It may also be the case that the cost increases with respect to the magnitude of the initial condition.  Another purpose is to check which scenario gives the least total endpoints $\lVert x(T;\bar{u})\rVert_1$ and if the optimal control results in positive value for $\lVert x(0;\bar{u})\rVert_1-\lVert x(T;\bar{u})\rVert_1$, then there is a finite number $n>1$ such that the continuation of optimal control strategy will completely eradicate mosquito population in $nT$ days. 

Regarding the definition of the total cost $C(u)$, we define a weighting factor $A>0$ such that
\begin{equation}
C(u):=\frac{A}{2T}\int_{0}^{T}\sum_{i=1}^{2}\omega_{x,i}u_i^2\,\dd t.
\end{equation} 

Numerical results from \ref{sce:1bb} are given as in Figs.~\ref{fig:ineggbb}-\ref{fig:controlbb}.
\begin{figure}[htb]
\begin{minipage}{.495\textwidth}
\includegraphics[width=\textwidth]{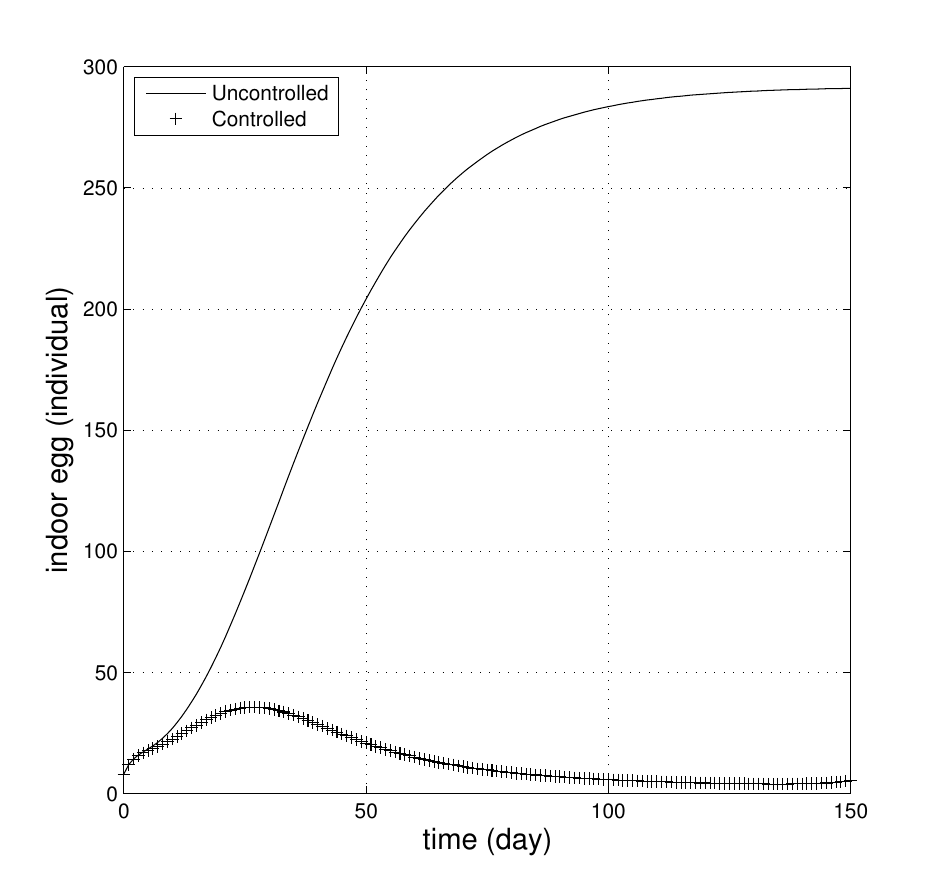}
\caption{\label{fig:ineggbb}Trajectory of indoor egg compartment.}
\end{minipage}
\hfill
\begin{minipage}{.495\textwidth}
\includegraphics[width=\textwidth]{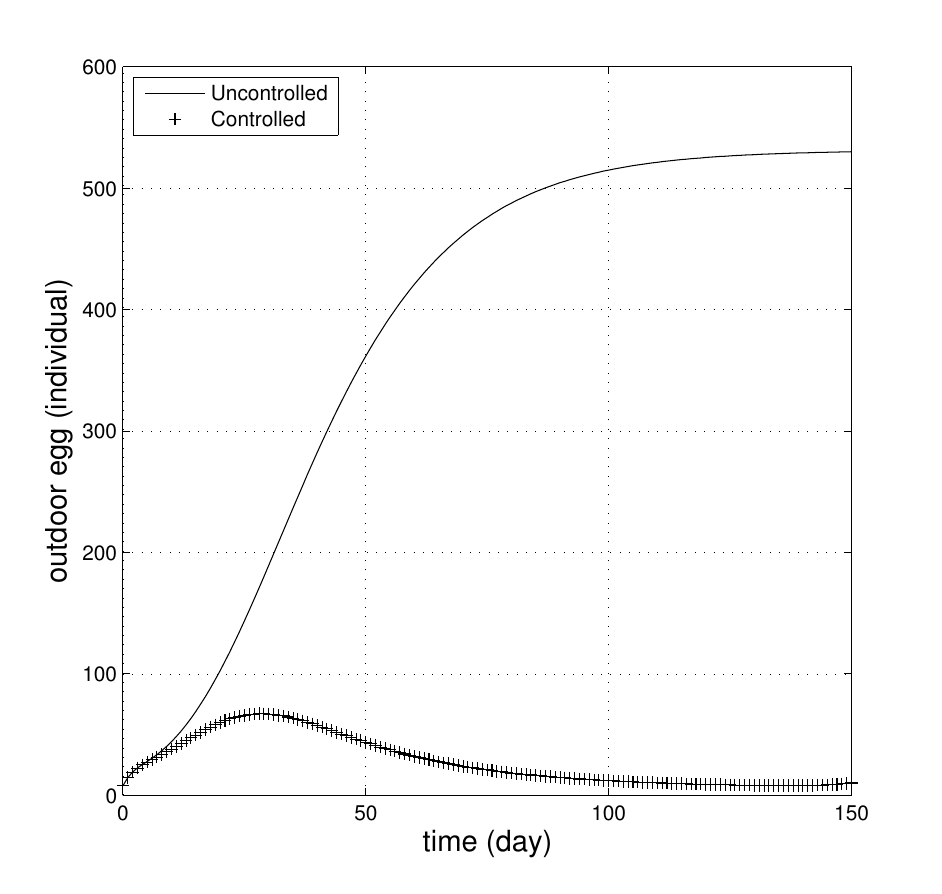}
\caption{\label{fig:outeggbb}Trajectory of outdoor egg compartment.}
\end{minipage}
\end{figure}

\begin{figure}[htb]
\begin{minipage}{.495\textwidth}
\includegraphics[width=\textwidth]{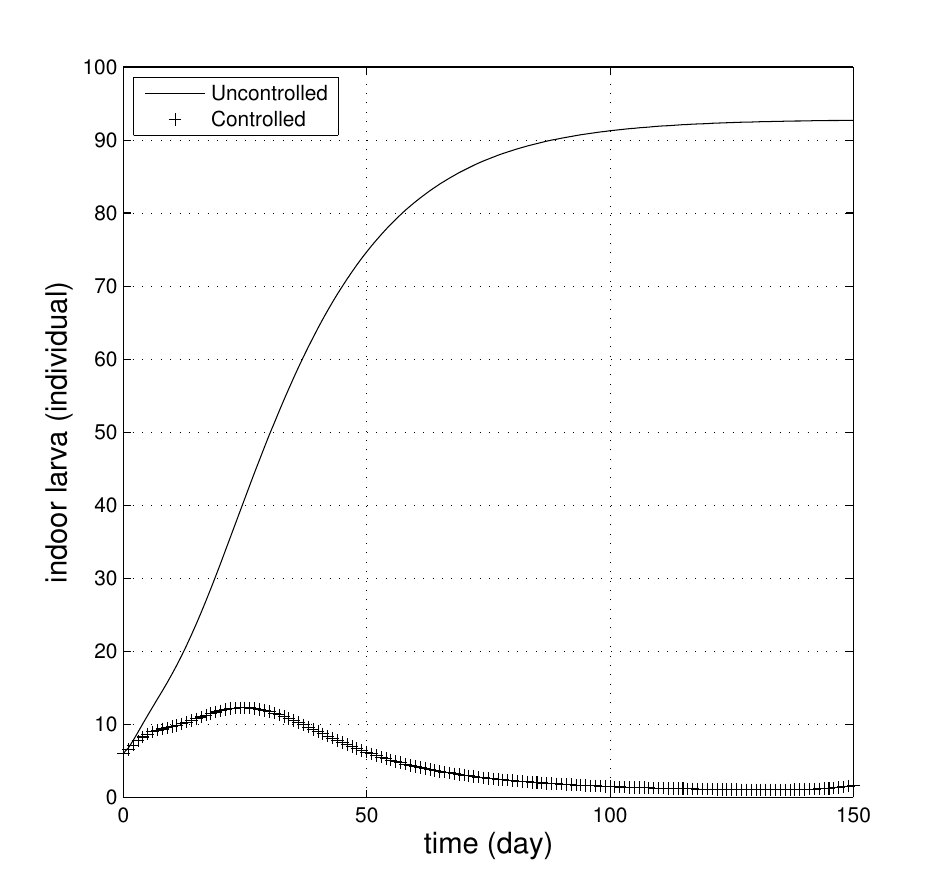}
\caption{\label{fig:inlarvbb}Trajectory of indoor larva compartment.}
\end{minipage}
\hfill
\begin{minipage}{.495\textwidth}
\includegraphics[width=\textwidth]{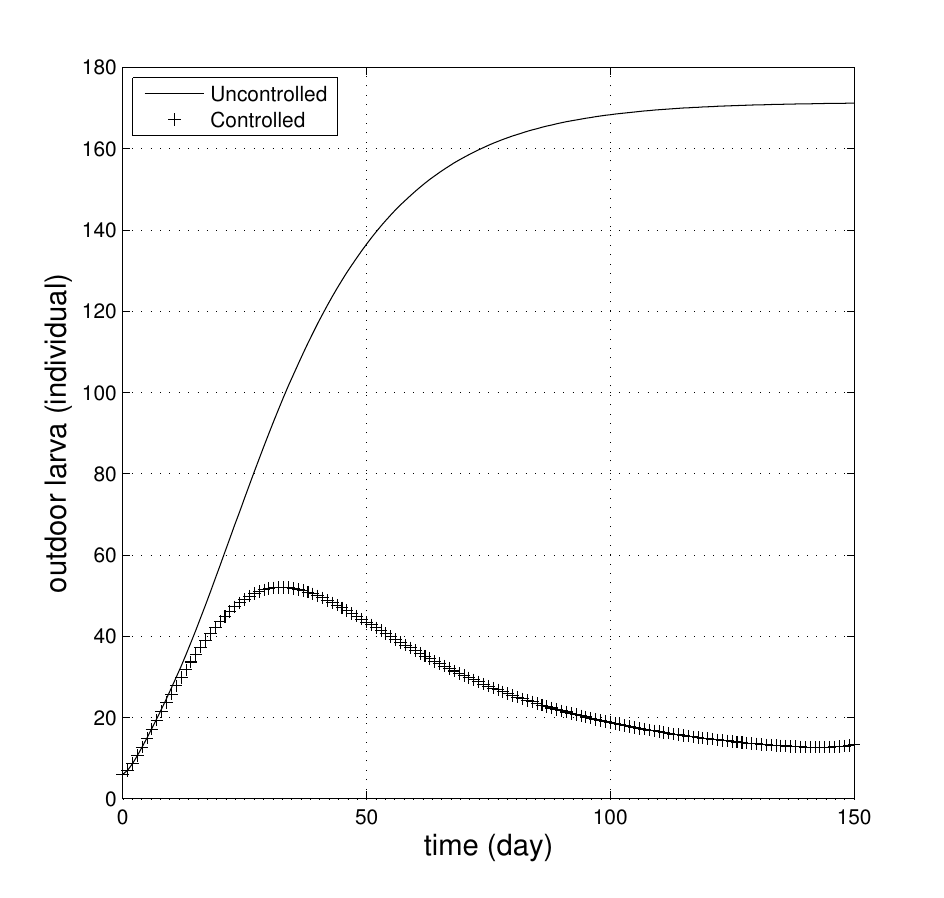}
\caption{\label{fig:outlarvbb}Trajectory of outdoor larva compartment.}
\end{minipage}
\end{figure}

\begin{figure}[htb]
\begin{minipage}{.495\textwidth}
\includegraphics[width=\textwidth]{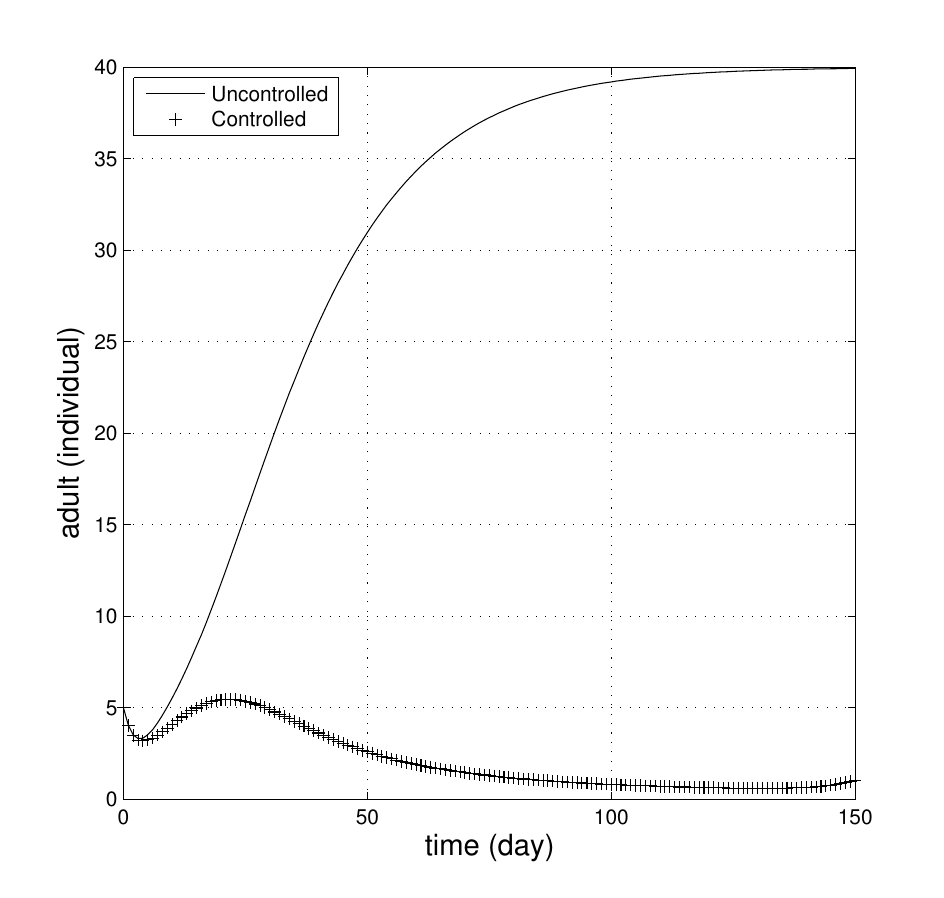}
\caption{\label{fig:adultbb}Trajectory of adult compartment.}
\end{minipage}
\hfill
\begin{minipage}{.495\textwidth}
\includegraphics[width=\textwidth]{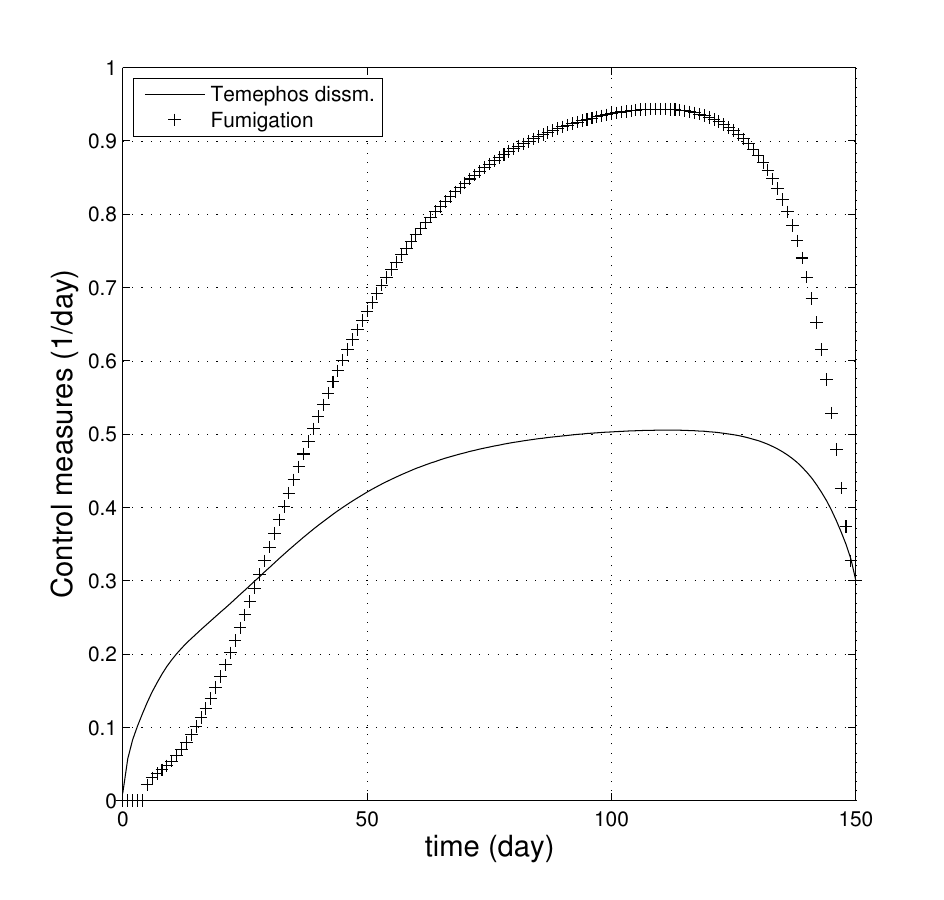}
\caption{\label{fig:controlbb}Trajectory of optimal control.}
\end{minipage}
\end{figure}

At a first look, it is clear from Figs.~\ref{fig:ineggbb}-\ref{fig:adultbb} that the optimal control makes significant reductions of all compartments. Similar results can be found from both \ref{sce:2bb} and \ref{sce:3bb}. It turns out that by using the data from Table \ref{tab:1}, fumigation is preferred over temephos during the application of the control. Figs.~\ref{fig:rubb} and~\ref{fig:controlbb} tell us so. For further managerial reference, Table \ref{tab2} represents the performance of the optimal control strategy within all scenarios. We also highlight the optimal control performance when one scheme vanishes.
\begin{table*}[h!]
\centering
\begin{tabular}{l|l|l|l}
\hline
$C(u)$, with $A=1$ & \ref{sce:1bb} & \ref{sce:2bb} & \ref{sce:3bb} \\ \hline
$u_1\neq 0,\,u_2=0$ & $7.4216\times 10^2$ & $2.9754\times10^{2}$ & $2.2304\times10^{2}$ \\ 
$u_1= 0,\,u_2\neq 0$ & $6.6658\times 10^3$ & $9.5861\times 10^3$ & $1.0150\times10^{4}$\\ $u_1\neq 0,\,u_2\neq 0$ & $1.4004\times10^{4}$ & $1.6217\times10^{4}$ & $1.6329\times10^{4}$ \\ \hline \hline
$\lVert x(0;\bar{u})\rVert_1$, $\lVert x(T;\bar{u})\rVert_1$ & \ref{sce:1bb} & \ref{sce:2bb} & \ref{sce:3bb} \\ \hline
$u_1\neq 0,\,u_2=0$ & 33, 839.57 & 794.50, 889.13 & $1.09\times10^{3}$, 901.57\\ 
$u_1= 0,\,u_2\neq 0$ & 33, 157.66 & 794.50, 154.73 & $1.09\times10^{3}$, 156.23\\ 
$u_1\neq 0,\,u_2\neq 0$ & 33, 31.62 & 794.50, 35.34 & $1.09\times10^{3}$, 37.52 \\ \hline \hline
$\frac{1}{T}\sum_{i=1}^{5}\lVert x_i(t;\bar{u})\rVert_{L^1}$ & \ref{sce:1bb} & \ref{sce:2bb}  & \ref{sce:3bb} \\ \hline
$u_1\neq 0,\,u_2=0$ & $610.9374$ & $946.5915$ & $1.0169\times 10^3$ \\ 
$u_1= 0,\,u_2\neq 0$ & $196.5743$ & $225.8798$ & $235.1042$\\ 
$u_1\neq 0,\,u_2\neq 0$ & $74.7872$ & $119.5978$ & $132.6077$ \\ \hline
\end{tabular}
\caption{\label{tab2}The performance of optimal control.} 
\end{table*} 

From the table, we conclude:
\begin{enumerate}
\item A complete control intervention using the two schemes appears with the highest cost for each of the three tested scenarios compared to the ones when one scheme vanishes. It is also noted that the cost slightly increases from \ref{sce:1bb} to \ref{sce:3bb}, meaning there is a positive correlation between the initial condition taken for simulation and the cost. If we consider the epidemiology of dengue amongst humans and mosquitoes, this fact can help us to decide when we should start to conduct an eradication programme. It appears as well in the table that there is a notable difference between the cost for maintaining only temephos dissemination and only fumigation. Mathematically, this phenomenon arises due to the choice of the trade-off coefficients $\omega_{x}$. In this case, we assume that the adults need to be accounted for more attention, hence we put a higher value for $\omega_{x,5}$. However, this does not always mean that taking more consideration into the reduction of adults requires a higher measure of fumigation rather than that of temephos. In our simulation, even if we take $\omega_{x,5}$ similar to the others, the result is more or less similar to the one as depicted in Fig.~\ref{fig:controlbb}. The reader can play around with these trade-off coefficients in order to get the result that suits the real situation best.
\item From the table, it is seen that a complete combination of the two control schemes can extremely reduce the total endpoints. Deyploying only temephos leads us to the worst case, i.e. there will be up to 800 total individuals remaining alive after $T$ days of treatment. Meanwhile, deploying only fumigation results in moderate total endpoints. If we compute $\lVert x(0;\bar{u})\rVert_1-\lVert x(T;\bar{u})\rVert_1$, then we would highly recommend to apply both control schemes as an integrated programme. 
\item We interpret $\frac{1}{T}\sum_{i=1}^{5}\lVert x_i(t;\bar{u})\rVert_{L^1}$ as the average number of the total population size at each time during observation. The higher its value, the higher the number of dengue incidences that may occur. Again, deploying only temephos leads to the worst case. We highlight that the complete control with two schemes gives the best reduction. Now from the table we know that a preventive act of starting the control programme when the number of individuals is small and a combination of both control schemes will bring us the best results: least cost, fewest total endpoints and fewest number of individuals living during the treatment. However, due to practical limitations, this ideal situation may not be achieved. Knowing the average density of mosquito population in a house in an endemic region, we can derive a similar results to those in Table \ref{tab2} that lead to  the best decision for a suitable control management.
\end{enumerate}

\section{Conclusion}

We have an optimal control model of mosquito population dynamics with indoor-outdoor life cycle classification. This model captures the situation in most dengue endemic regions where mosquitoes originate from indoor as well as outdoor. Two schematic control measures are added into the model as regulators that reduce the evolution of mosquito population in forward time. The first measure represents the rate of use of temephos, a chemical that can kill larvae and a small percentage of eggs in some indoor spots. The second measure represents the rate of use of fumigation, which targets adults. Some brief underlying results from the work in this paper are given as follows.

In the constant control case, the proposed model enumerates several biological meaningfulnesses. First, the evolution of mosquito population in forward time results in positive values. Second, with the given estimate of parameters and initial condition, it can numerically be seen that the evolution of total population remains bounded for all $t>0$ by the positive number $\max\{\lVert x(0;u)\rVert_1,\lVert Q_4\rVert_1\}$. Third, if the basic mosquito offspring number is less than one, then two conditions arise: (i) the zero equilibrium is locally asymptotically stable and (ii) a coexistence equilibrium does not exist. If the basic mosquito offspring number is greater than one, then three conditions arise: (i) the zero equilibrium is unstable, (ii) a coexistence equilibrium exists and (iii) is globally asymptotically stable in the nonnegative orthant.

The optimal control results as displayed in Table \ref{tab2} can be used as a reference in a decision-making process. A brief conclusion states that, in an endemic region, the best mosquito control impacts are produced if the control starts when the number of individuals is as small as possible, also with a combination of the two simultaneous schemes proposed in this paper. Both in constant control and optimal control cases, the implementation of fumigation is preferred over that of temephos. 

\section*{Acknowledgements}
The first and third authors acknowledge the financial support from Indonesia Endowment Fund for Education (LPDP) and Indonesia Directorate of Higher Education (DIKTI) on behalf of ITB Research Grant 2011. The authors are grateful to Prof. Mick Roberts (Massey University, NZ) and Prof. Neville Fowkes (University of Western Australia) and also the handling editor for their constructive comments and recommendations.


\end{document}